\documentclass[conference]{IEEEtran}
\usepackage{times,epsfig,graphicx,wrapfig,color,bbm,multirow,euscript,float}
\usepackage{amsmath,amssymb,dsfont,epsf,graphics,float,nicefrac,mathrsfs}

\newtheorem{lemma}{Lemma}
\newtheorem{theorem}{Theorem}
\newtheorem{proposition}{Proposition}

\def\wt{\text{wt}}

\newcommand\nc\newcommand
\nc\bfa{{\boldsymbol a}}\nc\bfA{{\bf A}}\nc\cA{{\mathcal A}}
\nc\bfb{{\boldsymbol b}}\nc\bfB{{\boldsymbol B}}\nc\cB{{\mathcal B}}
\nc\bfc{{\boldsymbol c}}\nc\bfC{{\bf C}}\nc\cC{{\mathcal C}}
\nc\sC{{\mathscr C}}
\nc\bfd{{\boldsymbol d}}\nc\bfD{{\bfD}}
\nc\cD{{\mathcal D}}
\nc\bfe{{\boldsymbol e}}\nc\bfE{{\bf E}}\nc\cE{{\mathcal E}}
\nc\bff{{\boldsymbol f}}\nc\bfF{{\bf F}}\nc\cF{{\mathcal F}}
\nc\bfg{{\boldsymbol g}}\nc\bfG{{\bf G}}\nc\cG{{\mathcal G}}
\nc\bfh{{\boldsymbol h}}\nc\bfH{{\bf H}}\nc\cH{{\mathcal H}}
\nc\bfi{{\boldsymbol i}}\nc\bfI{{\bf I}}\nc\cI{{\mathcal I}}\nc\sI{{\mathscr I}}
\nc\bfj{{\boldsymbolj}}\nc\bfJ{{\bf J}}\nc\cJ{{\mathcal J}}
\nc\bfk{{\boldsymbolk}}\nc\bfK{{\bf K}}\nc\cK{{\mathcal K}}
\nc\bfl{{\boldsymboll}}\nc\bfL{{\bf L}}\nc\cL{{\mathcal L}}
\nc\bfm{{\boldsymbolm}}\nc\bfM{{\bf M}}\nc\cM{{\mathcal M}}
\nc\bfn{{\boldsymboln}}\nc\bfN{{\bf N}}\nc\cN{{\mathcal N}}
\nc\bfo{{\boldsymbolo}}\nc\bfO{{\bf O}}\nc\cO{{\mathcal O}}
\nc\bfp{{\boldsymbolp}}\nc\bfP{{\bf P}}\nc\cP{{\mathcal P}}
\nc\eP{{\EuScriptP}}\nc\fP{{\mathfrak P}}
\nc\bfq{{\boldsymbol q}}\nc\bfQ{{\bf Q}}\nc\cQ{{\mathcal Q}}
\nc\bfr{{\boldsymbol r}}\nc\bfR{{\bf R}}\nc\cR{{\mathcal R}}
\nc\bfs{{\boldsymbol s}}\nc\bfS{{\boldsymbol S}}\nc\cS{{\mathcal S}}
\nc\bft{{\boldsymbol t}}\nc\bfT{{\bf T}}\nc\cT{{\mathcal T}}
\nc\bfu{{\boldsymbol u}}\nc\bfU{{\bf U}}\nc\cU{{\mathcal U}}
\nc\bfv{{\boldsymbol v}}\nc\bfV{{\bf V}}\nc\cV{{\mathcal V}}
\nc\bfw{{\boldsymbol w}}\nc\bfW{{\bf W}}\nc\cW{{\mathcal W}}
\nc\bfx{{\boldsymbol x}}\nc\bfX{{\bf X}}\nc\cX{{\mathcal X}}
\nc\bfy{{\boldsymbol y}}\nc\bfY{{\bf Y}}\nc\cY{{\mathcal Y}}
\nc\bfz{{\boldsymbol z}}\nc\bfZ{{\bf Z}}\nc\cZ{{\mathcal Z}}

\nc\od{{\bar d}}\nc\ow{{\barw}}\nc\odelta{{\bar\delta}}
\nc\ox{{\bar x}}\nc\oy{{\bary}}\nc\ou{{\bar u}} \nc\oh{{\bar h}}

\newcommand\ff{{\mathbb F}}

\newcommand\integers{{\mathbb Z}}
\nc\dgv{\delta_{\text{\rm GV}}} \nc\dcrit{\delta_{\text{\rm{crit}}}}
\nc\Esp{E_{\text{\rm sp}}}
\renewcommand\epsilon{\varepsilon}
\nc\eps{\varepsilon} \nc\ueps{\underline{\eps}}
\nc\ut{\underline{\theta}} \nc\ugam{\underline{\gamma}}
\nc\hr{\overrightarrow{H}} \nc\hl{\overleftarrow{H}}

\nc\rp{P}
\nc\lp{{P}^\bot}

\newcommand{\beeq}{\begin{eqnarray*}}
\newcommand{\eneq}{\end{eqnarray*}}

\newcommand{\remove}[1]{}

\newcommand{\half}{\nicefrac12}

\title{\Large\bf Polar codes for $q$-ary channels, $q=2^r$}

\author{\IEEEauthorblockN{Woomyoung Park}
\IEEEauthorblockA{Department of ECE and Inst. Sys. Res.\\
University of Maryland\\
College Park, Maryland, 20742\\
Email: woomyoung.park@gmail.com}

\and \IEEEauthorblockN{Alexander Barg}
\IEEEauthorblockA{Department of ECE and Inst. Sys. Res.\\
University of Maryland\\
College Park, Maryland, 20742\\
Email: abarg@umd.edu} }

\begin{document}\maketitle
\thispagestyle{empty}

\begin{abstract}
We study polarization for nonbinary channels with input alphabet of size
$q=2^r,r=2,3,\dots$. Using Ar\i kan's polarizing kernel $H_2$, we prove
that the virtual channels that arise in the process of polarization converge
to $q$-ary channels with capacity $1,2,\dots,r$ bits, and that the total transmission
rate approaches the symmetric capacity of the channel. This leads to an explicit
transmission scheme for $q$-ary channels.
The error probability of decoding using successive cancellation
behaves as $\exp(-N^\alpha),$ where $N$ is the code length and $\alpha$ is any
constant less than $0.5.$
\end{abstract}


\section{Introduction}
Polarization is a new concept in information theory discovered in
the context of capacity-achieving families of codes for symmetric
memoryless channels and later generalized to source coding,
multi-user channels and other problems. Polarization was first
described by Ar{\i}kan \cite{ari09} who constructed binary codes
that achieve capacity of symmetric memoryless channels (and
``symmetric capacity'' of general binary-input channels).
The main idea of \cite{ari09} is to combine the bits of the source
sequence using repeated application of the ``polarization kernel''
{\small   $H_2=\Big(\hspace*{-.05in}\begin{array}{c@{\hspace*{0.05in}}c}
    1&0\\1&1\end{array}\hspace*{-.05in}\Big).$}
The resulting linear code of length $N=2^n$ has the generator matrix which
forms a submatrix of $G_N=B H_2^{\otimes n},$ where $B$ is a permutation matrix.
The choice of the rows of $G_N$
is governed by the polarization of virtual channels for individual bits that
arise in the process of channel combining and splitting. Namely, the data
bits are written in the coordinates that correspond to near-perfect channels
while the other bits are fixed to some values known to both the transmitter and
the decoder. It was shown later that polarization on binary
channels can be achieved using a variety of other kernels: in particular, any
$m\times m$ matrix whose columns cannot be arranged to form an upper triangular matrix,
achieves the desired polarization \cite{Korada2010}.

A study of polar codes for channels with nonbinary input was undertaken by
{\c S}a{\c s}o{\u g}lu et al. \cite{sas09,sas09a} and Mori and Tanaka \cite{mor10}.
For prime $q$, it suffices to take the kernel $H_2,$ while for nonprime
alphabets, the kernel is time-varying and not explicit. Namely, for prime $q,$
\cite{sas09} showed that there exist permutations of the input
alphabet such that the virtual channels for
individual $q$-ary symbols become either fully noisy or perfect, and
the proportion of perfect channels approaches the symmetric capacity,
in analogy with the results for binary codes in \cite{ari09}.
At the same time,
\cite{sas09} remarks that the transmission scheme that uses the kernel
$H_2$ with modulo-$q$ addition for composite $q$
does not necessarily lead to the polarization of the channels to the
two extremes. Rather, they show that there exists a sequence of permutations of
the input alphabet such that when they are combined with $H_2,$
the virtual channels for the transmitted symbols become either nearly perfect or
nearly useless.

The authors of \cite{sas09} suggest several alternatives to the kernel
$H_2$ that rely on randomized permutations or, in the case of $q=2^r,$
on multilevel schemes
that implement polar coding for each of the bits of the
symbol independently, combining them in the decoding procedure; see
esp. \cite{sas09a}.

In this paper we study polarization for channels with input alphabet of size $q=2^r,
r=2,3,\dots.$ Suppose that the channel is given by a stochastic matrix $W(y|x)$
where $x\in \cX, y\in\cY, \cX=\{0,1,\dots,q-1\},$ and $\cY$ is a finite alphabet. Assuming
that the channel combining is performed using the kernel $H_2$ with addition modulo
$q$, we establish results about the polarization of channels for individual
symbols. It turns out that virtual channels for the transmitted symbols converge
to one of $r+1$ {\em  extremal configurations} in which $j$ out of $r$ bits are
transmitted
near-perfectly while the remaining $r-j$ bits carry almost no information.  Moreover,
the good bits are always aligned to the right of the transmitted $r$-block, and no
other situations arise in the limit. Thus, the extremal configurations for information
rates that arise as a result of polarization are easily characterized: they form an
upper-triangular matrix as described in Sect.~\ref{sect:polar} (see also
Figs.~\ref{fig1},\,\ref{fig2} in the final section of the paper).
This characterization also constitutes the main difference of our results
from the multilevel scheme in \cite{sas09a}: there, the set of extremal configurations
can in principle have cardinality $2^r$ which complicates the code construction.

Another related work is the paper by Abbe and Telatar \cite{abb10a}. In it, the
authors observed multilevel polarization in a somewhat different context.  The main
result of their paper provides a characterization of extremal points of the region of
attainable rates when polar codes are used for each of the $r$ users of a
multiple-access channel.  Namely, as shown in \cite{abb10a} (see also \cite{abb10b}),
these points form a subset in the set of vertices of a matroid on the set of $r$
users.  \cite{abb10a} also remarks that these results translate directly to
transmission over a $q$-ary DMC, showing that the rate polarizes to many levels.  To
explain the difference between \cite{abb10a} and our work we note that transmission
over the multiple-access channel in \cite{abb10a} is set up in such a way that, once
applied to the DMC, it corresponds to encoding each bit of the $q$-ary symbol by its
own polar code (we again assume that $q=2^r$). In other words, the polarization
kernel employed is a linear operator $G=I_r\otimes H_2.$ Thus, the group acting on
$\cX$ is $\ff_{2^r}^+=\integers_2\times\dots\times\integers_2$ rather than the cyclic
additive group of order $q$ considered in this paper.

This work began as an attempt to construct polar codes for
the {\em ordered symmetric channel}, introduced in our earlier paper \cite{par11a}.
This channel provides an information-theoretic model related to the ordered distance
on binary $r$-vectors, defined as follows:
   \begin{equation}\label{eq:NRT}
     d_r(x,x')=\max\{j:x_j\ne x_j'\}, \quad\text{where }x,x'\in \{0,1\}^r.
  \end{equation}
Below $\wt_r(x)=d_r(x,0)$ denotes the ordered weight of the symbol $x$.
The ordered distance is an instance of a large class of metrics introduced in
\cite{bru95} following works of Niederreiter in numerical analysis \cite{nie91}. It
has subsequently appeared in a large number of works in algebraic combinatorics and
coding theory; see e.g., \cite{bar09b} and references therein.
We find it quite interesting that it independently arises
in the study of polar codes on channels with input of size $q=2^r.$
Examples of $q$-ary polar codes for
ordered symmetric channels can be easily constructed and analyzed.

Last but not least, when this work was in its final stages, we became
aware of the paper by Sahebi and Pradhan \cite{sah11} who also
observed the multilevel polarization phenomenon for $q$-ary
channels. At the same time, \cite{sah11} did not give a proof of
polarization, which constitutes the main technical part of our work.
The motivation of the approach of \cite{sah11} relates to a detailed
study of linear and group codes on $q$-ary channels, and is also
different from our approach.

In the next section we state and prove the main result, the convergence of the channels
to one of the $r+1$ extremal configurations, and deduce that polar codes achieve the symmetric
capacity of the channel. Then we derive the rate of polarization and estimate the error probability
of decoding, and give some examples.

\section{Polarization for $q$-ary channels}
We consider combining of the $q$-ary data under the action of the operator $H_2,$
where $q=2^r, r\ge 2.$
Let $W: \cX\to \cY, |\cX|=q$ be a discrete memoryless channel (DMC).
The {\em symmetric capacity} of the channel $W$ equals
\begin{align*}
I(W) \triangleq \sum_{x \in \cX}\sum_{y\in\cY} \frac{1}{q}
W(y|x)\log{\frac{W(y|x)}{\sum_{x'\in\cX}\frac{1}{q}W(y|x')}}
\end{align*}
where the base of the logarithm is $2$.
Define the combined channel $W_2$ and the channels $W^-$ and $W^+$ by
\begin{align}
W_2(y_1, y_2|u_1, u_2)
&= W(y_1|u_1 + u_2) W(y_2|u_2),\nonumber\\
W^-(y_1,y_2|u_1) &= \sum_{u_2 \in \cX} \frac{1}{q}
W_2(y_1, y_2|u_1, u_2),\label{eq:+-1}\\
W^+(y_1, y_2, u_1| u_2) &= \frac{1}{q} W_2(y_1, y_2 |u_1, u_2),
\label{eq:+-2}\end{align}
where $u_1, u_2, y_1, y_2$ are $r$-vectors and $+$ is a
modulo-$q$ sum. This transformation can be applied recursively to the channels
$W^-,W^+$ resulting in four channels of the form $W^{b_1b_2}, b_1,b_2\in\{+.-\}$. After $n$
steps we obtain $N=2^n$ channels $W_N^{(j)}, j=1,\dots, N.$
For the case $q=2$ it is shown in \cite{ari09} that as $n$ increases, the channels $W_N^{(j)}$
become either almost perfect or almost completely noisy (polarize). In formal terms,
for any $\epsilon>0$
   \begin{equation}\label{eq:binary}
    \lim_{n\to\infty}\frac{|\{b\in\{+,-\}^n: I(W^b)\in(\epsilon,1-\epsilon)\}|}{2^n}=0.
   \end{equation}
In this paper we extend this result to the case $q=2^r, r>1$.

As shown
in \cite{ari09}, after $n$ steps of the transformation \eqref{eq:+-1}-\eqref{eq:+-2}
the channels $W_N^{(i)}:\cX\to\cY^N\times\cX^{i-1}, 1\le i\le N$ are given by
  \begin{equation}\label{eq:WNi}
    W_N^{(i)}(y_1^N,u_1^{i-1}|u_i)=\frac1{q^{N-1}}\sum_{u_{i+1}^N\in \cX^{N-i}}W^N(y_1^N|u_1^NG_N),
  \end{equation}
where $G_N=BH_2^{\otimes n}$ and $B$ is a permutation matrix. Here we use the shorthand notation
for sequences of symbols: for instance, $y_1^N\triangleq (y_1,y_2,\dots,y_N),$ etc.

\subsection{Notation}

For any pair of input symbols
$x,x' \in \cX$, the Bhattacharyya distance
between them is
\begin{align*}
Z(W_{\{x,x'\}}) = \sum_{y \in \cY}\sqrt{W(y|x)W(y|x')}
\end{align*}
where $W_{\{x,x'\}}$ is the channel obtained by restricting the input
alphabet of $W$ to the subset $\{x,x'\} \subset \cX$.

Define the quantity $Z_{v}(W)$ for $v \in \cX\setminus\{0\}$:
  $$
Z_{v}(W) = \frac{1}{2^r} \sum_{x\in\cX}Z(W_{\{x,x+v\}}).
  $$
Introduce the $i$th average Bhattacharyya distance of the channel $W$ by
  \begin{equation}\label{eq:zi}
Z_i(W) = \frac{1}{2^{i-1}}\sum_{v\in\cX_i}Z_{v}(W)
  \end{equation}
where $i = 1,2,\cdots,r$ and $\cX_i=\{v\in \cX:\wt_r(v)=i\}.$ Then
\begin{align}\label{eq:ZW}
  Z(W):&=\frac{1}{2^r(2^r-1)}\sum_{x\ne x'}
  Z(W_{\{x,x'\}})\nonumber\\
  &=\frac1{2^r-1} \sum_{i=1}^r 2^{i-1}Z_i(W)
\end{align}

Recall the setting of \cite{ari09} for the evolution of the channel parameters.
On the set $\Omega=\{+,-\}^\ast$ of semi-infinite binary sequences
define a $\sigma$-algebra $\cF$ on $\Omega$ generated by the cylinder sets
$S(b_1,\dots,b_n)=\{\omega\in \Omega:\omega_1=b_1,\dots,\omega_n=b_n\}$
for all sequences $(b_1,\dots,b_n)\in\{+,-\}^n$ and for all $n\ge 0.$
Consider the probability space $(\Omega,\cF,P),$ where
$P(S(b_1,\dots,b_n))=2^{-n}, n\ge 0.$ Define a filtration
$\cF_0\subset \cF_1\subset\dots\subset \cF$ where $\cF_0=\{\emptyset,\Omega\}$
and $\cF_n, n\ge 1$ is generated by the cylinder sets $S(b_1,\dots,b_n), b_i\in\{+,-\}.$

Let $B_i,i=1,2,\cdots$ be i.i.d. $\{+,-\}$-valued random variables with
$\Pr(B_1 = +) = \Pr(B_1 = -) = 1/2.$ The random channel emerging at time $n$ will
be denoted by $W^B,$ where $B=(B_1,B_2,\cdots,B_n).$ Thus, $P(W^B=W_N^{(i)})=2^{-n}$
for all $i=1,\dots,2^n.$ Let $W_n=W^B,$ $I_n=I(W^B),$
$Z_{\{x,x'\},n}=Z(W_{\{x,x'\}}^B),$ $Z_{v,n}=Z_{v}(W^B),$ and $Z_{i,n}=Z_i(W^B).$
These random variables are adapted to the above filtration (meaning that $I_n$ etc.
are measurable w.r.t. $\cF_n$ for every $n\ge 1$).

\subsection{Channel polarization}\label{sect:polar}
In this section we state a sequence of results that shows that $q$-ary polar codes
based on the kernel $H_2$ can be used to transmit
reliably over the channel $W$ for all rates $R<I(W).$

\vspace*{.05in}\begin{theorem}\label{thm:main} {\sl
(a) Let $n\to \infty.$ The random variable $I_n$ converges a.e. to a random variable
$I_\infty$ with $E(I_\infty)=I(W).$

(b) For all $i=1,2,\dots,r$
  $$
   \lim_{n\to\infty}Z_{i,n}= Z_{i,\infty} \quad a.e.,
  $$
where the variables $Z_{i,\infty}$ take values $0$ and $1.$
With probability one the vector $(Z_{i,\infty}, i=1,\dots,r) $ takes one
of the following values:
  \begin{equation}\label{eq:triangular}
\begin{array}{*5{c@{\hspace*{0.03in}}}}
(Z_{1,\infty}=0,&Z_{2,\infty}=0,&\dots,&Z_{r-1,\infty}=0,&Z_{r,\infty}=0)\\
(Z_{1,\infty}=1,&Z_{2,\infty}=0,&\dots,&Z_{r-1,\infty}=0,&Z_{r,\infty}=0)\\
(Z_{1,\infty}=1,&Z_{2,\infty}=1,&\dots,&Z_{r-1,\infty}=0,&Z_{r,\infty}=0)\\
\vdots&&\vdots&&\vdots\\
(Z_{1,\infty}=1,&Z_{2,\infty}=1,&\dots,&Z_{r-1,\infty}=1,&Z_{r,\infty}=0)\\
(Z_{1,\infty}=1,&Z_{2,\infty}=1,&\dots,&Z_{r-1,\infty}=1,&Z_{r,\infty}=1).
\end{array}
  \end{equation}

\remove{(c) Let $\epsilon>0$ and $D=[0,\epsilon)\cup(\bigcup\limits_{k=1}^{r-1}(k-\epsilon,k+\epsilon))
   \cup (r-\epsilon,r].$ Then
   $$
  \lim_{n\to\infty}\frac{|\{i:I(W_N^{(i)})\in D\}|}{2^n}=1.
   $$ }
}
\end{theorem}

\vspace*{.05in} Let us restate part (b) of this theorem for finite $n$.
\begin{proposition}\label{prop:finite}
{\sl Let $\epsilon,\delta>0$ be fixed.
For $k=0,1,\dots,r$ define disjoint events
  $$
B_{k,n}(\epsilon)=\Big\{\omega: (Z_{1,n},Z_{2,n},\dots,Z_{r,n})\in \cR_k\Big\}
  $$
where $\cR_k=\cR_k(\epsilon)\triangleq
\Big(\prod_{i=1}^{k}D_1\Big)\times \Big(\prod_{i=k+1}^r D_0\Big)
  $ and
$D_0=[0,\epsilon),$ $D_1=(1-\epsilon,1].$
Then $P(\cup_{k=0}^{r} B_{k,n}(\epsilon))\ge 1-\delta$ starting from some $n=n(\epsilon,\delta).$ }
\end{proposition}
\vspace*{.05in} The proofs of these statements are given in a later part of this section.

\vspace*{.05in}
We need the following lemma.

\begin{lemma}\label{lemma:iwzw} For a DMC with $q$-ary input, $I(W)$ and
$Z(W)$ are related by
\begin{align}
I(W) &\geq \log{\frac{2^r}{1 + \sum_{i=1}^{r} 2^{i-1}Z_i(W)}}\label{eq:lower}\\
I(W) &\leq \sum_{i=1}^{r} \sqrt{1-Z_i(W)^2}.\label{eq:upper}
\end{align}
\end{lemma}

For $r=1$ these inequalities are proved in \cite{ari09}.
For $r>1$ Eq.~\eqref{eq:lower} is a restatement of \cite[Prop.~3]{sas09} using \eqref{eq:ZW}.
The fact that \eqref{eq:upper} holds for all $r>1$ is new, and is proved in the Appendix.

Inequalities \eqref{eq:lower}-\eqref{eq:upper} imply that if
$(Z_1,\dots,Z_r)\in \cR_k(\epsilon)$ then $|I(W)-(r-k)|\le \delta$ where
  $
  \delta\ge\max(k\sqrt\epsilon,(2^{r-k}-1)\epsilon\log e).
  $

The following proposition is an immediate corollary of the above results.
\begin{proposition}\label{prop:Ik}{\sl
(a) The random variable $I_\infty$ is supported on the set $\{0,1,\dots,r\}.$

(b) For every $0\le k\le r$ and every $\delta>0$
there exists $\epsilon>0$ such that
   $$
    \lim_{n\to\infty}P(\{|I_n-(r-k)|\le\delta\} \bigtriangleup
 B_{k,n}(\epsilon))
  =0.
   $$

(c) $E(|\{i:Z_{i,\infty}=0\}|)=I(W).$}
\end{proposition}
\begin{proof} The first statement is obvious from \eqref{eq:lower}-\eqref{eq:upper}.
To prove the second statement we note that, with the appropriate choice of $\epsilon$
  $$
   \{|I_n-(r-k)|\le\delta\}\supset  B_{k,n}(\epsilon)
  $$
for all $n\ge 0.$ At the same time,
  $
   P(\{|I_n-(r-k)|\le\delta\}\cap B_{k',n}(\epsilon))=0
  $
for all $k'\ne k,$ and $P(\stackrel{\circ}{\cup} B_{k,n}(\epsilon))\to 1$ for any $\epsilon>0.$
Together this implies (b).
Finally, we have that $E(I_\infty)=I(W).$ Then use (a) and (b) to claim that
   $
E(|\{i:Z_{i,\infty}=0\}|)=\sum_{k=0}^{r} k P(I_{\infty}=k) =I(W).
   $
\end{proof}

\vspace*{.05in} We can say a bit more about the nature of convergence established in this
proposition.
Let us fix $k\in\{0,1,\dots,r\}$ and
define the channel for the $r-k$ rightmost bits of the
transmitted symbol as follows:
\begin{align*}
    W^{[r-k]}(y|u) = \frac1{2^k}\sum_{x\in \cX: x_{k+1}^r=u}
    W(y|x), \qquad u \in \{0,1\}^{r-k}
\end{align*}
where $x = (x_1,x_2,\dots,x_r)$.
\vspace*{.05in}\begin{lemma} {\sl Let $V:\cX\to \tilde \cY$ be a DMC and let $\delta>0.$
Suppose that
$(Z_{1,n}(V),Z_{2,n}(V),\dots,Z_{r,n}(V))\in\cR_k(\epsilon),$
for some $0\le k\le r.$ If $\epsilon$ is sufficiently small, then
$I(V^{[r-k]})\ge r-k-\delta.$ In particular, it suffices to take
$\epsilon\le 2^{-k+\delta}/(2^{r-k}-1).$}
\end{lemma}
\begin{proof} We may assume that $1\le k\le r-1.$
Let $u\in \cX^{r-k}, x=(x_1,\dots,x_k,u)\in\cX, x'=(x_1',\dots,x_k',u)\in\cX.$
Let $v\in\{0,1\}^{r-k}\backslash\{0\}$ and consider
\begin{align*}
    Z(V_{\{u,u+v\}}^{[r-k]})
   &= \sum_y \sqrt{ V^{[r-k]}(y|u)  V^{[r-k]}(y|u+v)}\\
&=\frac1{2^k}\sum_y\sqrt{\sum_{x}\sum_{x'}
    V(y|x)V(y|x'+v')}\\
&\leq \frac1{2^k}\sum_y\sum_{x}\sum_{x'}\sqrt{V(y|x)V(y|x'+v')}\\
&=\frac1{2^k}\sum_{x,x'}Z(V_{\{x,x'+v'\}})\\
&< 2^{k} \epsilon
\end{align*}
where $v' =0^kv_1v_2\dots v_{r-k}.$ The last inequality follows from the fact
that $Z_i(V) < \epsilon$ for $i = k+1,\dots,r.$  Since
$Z_i(V^{[r-k]})$ is the average of the $Z(V_{\{u,u+v\}}^{[r-k]})$ over all
$v$ with $\wt_r(v) = i$, $Z_i(V^{[r-k]}) < 2^{k} \epsilon$ for all
$i=1,\dots,r-k$. Now the lemma follows from
 \eqref{eq:lower} in Lemma \ref{lemma:iwzw},
\end{proof}

It turns out that the channels for
individual bits converge to either perfect or fully noisy channels.
If the channel for bit $j$ is perfect then the channels
for all bits $i, r\ge i>j$ are perfect. If the channel for bit $i$ is noisy then
the channels for all bits $j, 1\le j<i$ are noisy. The total number of near-perfect
bits approaches $I(W)$.
This is made formal in the next proposition.
\begin{proposition}\label{prop:bits} {\sl Let $\Omega_k=\{\omega:(Z_{1,\infty},Z_{2,\infty},\dots,Z_{r,\infty})
=1^{k}0^{r-k}\}, k=0,1,\dots, r.$ For every $\omega\in\Omega_k$
    $$
   \lim_{n\to\infty}|I_n-I(W_n^{[r-k]})|=0.
    $$}
\end{proposition}
\begin{proof} For every $\omega\in\Omega_k$ we have that $I_n(\omega)\to r-k.$
Combining this with the previous lemma and Proposition \ref{prop:Ik}(b),
we conclude that for such $\omega$
also $I(W_n^{[r-k]})\to r-k.$ \end{proof}

\vspace*{.05in} The concluding claim of this section describes the channel polarization
and establishes that the total number of bits sent over almost noiseless channels
approaches $NI(W).$
\begin{theorem} {\sl For any DMC $W:\cX\to \cY$ the channels $W_N^{(i)}$ polarize to one
of the $r+1$ extremal configurations. Namely, let $V_i=W_N^{(i)}$ and
   $$
\pi_{k,N}=\frac{|\{i\in [N]: |I(V_i)-k|<\delta \wedge |I(V_i^{[k]})-k|<\delta\}|}{N},
   $$
where $\delta>0,$
then  $\lim_{N\to\infty}\pi_{k,N}=P(I_\infty= k)$ for all $k=0,1,\dots,r.$ Consequently
   $$
    \sum_{k=1}^r k \pi_k\to I(W).
   $$}
\end{theorem}

\vspace*{.05in}
This theorem follows directly from Theorem \ref{thm:main} and Propositions \ref{prop:Ik}
and \ref{prop:bits}.
Some examples of convergence to the extremal configurations described by this theorem are given
in Sect.~\ref{sect:ordered} below.

\subsection{Transmission with polar codes}
Let us describe a scheme of transmitting over the channel $W$
with polar codes. Take $\epsilon>0$ and choose a
sufficiently large $n.$ Assume that the length of the code is $N=2^n.$
Proposition \ref{prop:finite} implies that set $[N],$ apart from a small subset,
is partitioned into $r+1$ subsets $\cA_{k,n}$ such that for $j\in \cA_{k,n}$ the vector
$(Z_{1}(W_N^{(j)}),Z_{2}(W_N^{(j)}),\dots,Z_{r}(W_N^{(j)}))\in \cR_{k}(\epsilon).$
Each $j\in \cA_{k,n}$ refers to an $r$-bit symbol in which $r-k$ rightmost bits
correspond to small values of $Z_{i}(W_N^{(j)}).$ To transmit data over the channel, we write the data bits in these
coordinates and encode them using the linear transformation $G_N.$

More specifically, let us order the coordinates $j\in[N]$ by the increase
of the quantity $\sum_{i=1}^r 2^{i-1}Z_{i}(W_N^{(j)})$ and use these numbers to locate the
subsets $\cA_{k,n}.$
We transmit data by encoding messages $u_1^N=(u_1,\dots,u_N)$ in which if $j\in \cA_{k,n},
k=0,\dots,r-1$ then the symbol $u_j$ is taken from the subset of symbols
of $\cX$ with the first $k$ symbols fixed and known to  both the encoder and
the decoder (\cite{ari09} calls them frozen bits).  In particular, the subset $\cA_{r,n}$
is not used to transmit data.
A polar codeword is computed as $x_1^N=u_1^NG_N$ and sent over the channel.

Decoding is performed using the ``successive cancellation'' procedure
of \cite{ari09} with the obvious constraints on the symbol values. Namely, for
$j=1,\dots,N$ put
\begin{align*}
    \hat u_j =  \begin{cases}
                    u_j, & j \in \cA_{r,n} \\
                    \arg\max_{x} W_N^{(j)}(y_1^N,\hat u_1^{j-1}|x), & j\in \cup_{k\le r-1} \cA_{k,n}
                \end{cases}
\end{align*}
where if $j\in \cA_{k,n}, k=0,1,\dots,r-1,$ then the maximum is computed
over the symbols $x\in \cX$ with the fixed (known) values of the first
$k$ bits.

\remove{For $k=1,2\dots,r+1$ define
  $$
    I_{i,k}=\lim_{n\to\infty}\frac{|\{j: j\in\cA_{k,n} \wedge \hat u_{ij}=u_{ij}\}|}{|\cA_{k,n}|}.
  $$
\begin{theorem} \label{thm:decoding}
{\sl
Suppose that information is transmitted over a $q$-ary
DMC $W$, $q=2^r, r\ge 1$ using polar codes with polarization kernel $H_2$
and successive cancellation decoding.

(a) Then the vector $(I_{1,k},I_{2,k},\dots,I_{r,k})$ equals $0^{k-1}1^{r-k+1}$
($k-1$ zeros followed by $r-k+1$ ones) with probability one.
As a consequence, the random variable $I_{\infty}$ is supported on the set

(b) The number of bits per symbol transmitted reliably
over the channel approaches $I(W).$
\begin{equation}\label{eq:triangular}
\begin{array}{*5{c@{\hspace*{0.03in}}}}
(I_{1}=1,&I_{2}=1,&\dots,&I_{r-1}=1,&I_{r}=1)\\
(I_{1}=0,&I_{2}=1,&\dots,&I_{r-1}=1,&I_{r}=1)\\
(I_{1}=0,&I_{2}=0,&\dots,&I_{r-1}=1,&I_{r}=1)\\
\vdots&&\vdots&&\vdots\\
(I_{1}=0,&I_{2}=0,&\dots,&I_{r-1}=0,&I_{r}=1)\\
(I_{1}=0,&I_{2}=0,&\dots,&I_{r-1}=0,&I_{r}=0).
\end{array}
  \end{equation}}
\end{theorem}}
The error probability of this decoding is estimated in Sect.~\ref{sect:error}.

\subsection{Proof of Theorem \ref{thm:main}}
Part (a) of Theorem \ref {thm:main} follows straightforwardly from
\cite{ari09,sas09}.
Namely, as shown in \cite[Prop.~4]{ari09},
$I(W^+)+I(W^-)=2I(W).$ We note that the proof in \cite{ari09}
uses only the fact that $u_1,u_2$ are recoverable from $x_1,x_2$
which is true in our case.
Hence the sequence $I_n,n\ge 1$ forms a bounded martingale.
By Doob's theorem \cite[p.196]{kor07},
it converges a.e. in $L^1(\Omega,\cF,P)$ to a random variable
$I_\infty$ with $E(I_\infty)=I(W).$

To prove part (b) we show that each of the $Z_{i,n}$'s converges a.s. to a
$(0,1)$ Bernoulli random variable $Z_{i,\infty}$. This convergence occurs in
a concerted way in that the limit r.v.'s obey $Z_{j,\infty}\ge Z_{i,\infty}$ a.e.
if $j<i.$ This is shown by observing that for any fixed $i=1,\dots,r$ and
for all $v\in \cX_i$ , the $Z_{v,n}(W)$ converge to
identical copies of a Bernoulli random variable.

\vspace*{.1in} \subsubsection{Convergence of $Z_{v,n}, v\in \cX$}
\vspace*{.05in}
In this section we shall prove that the Bhattacharyya parameters
$Z_{v,n}$ converge almost surely to Bernoulli random variables.
The proof forms the main technical result of this paper and is accomplished in several steps.

\begin{lemma} Let
   $$
   Z_{\max}^{(j)}(W)=\max_{v\in \cX_j} Z_{v}(W), \quad j=1,\dots,r.
  $$
Then   \begin{equation}\label{eq:square}
   Z_{\max}^{(r-j)}(W^+)=Z_{\max}^{(r-j)}(W)^2, \quad j=0,\dots,r-1.
   \end{equation}
  \begin{equation}\label{eq:le0}
   Z_{\max}^{(r)}(W^-)\le qZ_{\max}^{(r)}(W)
 \end{equation}
  \begin{equation}\label{eq:le1}
   Z_{\max}^{(r-1)}(W^-)\le \frac q2Z_{\max}^{(r)}(W)+\frac q2Z_{\max}^{(r-1)}(W)
  \end{equation}
and generally
  \begin{multline}
   Z_{\max}^{(r-j)}(W^-)\le \frac q2 Z_{\max}^{(r)}(W)+\frac q4Z_{\max}^{(r-1)}(W)
  +\\ \dots+\frac q{2^j}Z_{\max}^{(r-j+1)}(W)+\frac q{2^j}Z_{\max}^{(r-j)}(W).
  \label{eq:le2}
     \end{multline}
\end{lemma}
\begin{proof}
In \cite{sas09} it is shown that for all $v\in \cX\backslash\{0\}$
  \begin{gather}
   Z_v(W^+)=Z_v(W)^2\label{eq:v-square}\\
  Z_v(W^{-})\le 2Z_v(W)+\sum_{\delta\in\cX\backslash\{0,-v\}}Z_\delta(W)Z_{v+\delta}(W).\label{eq:le}
  \end{gather}
The first of these two equations implies \eqref{eq:square}. Now
take $v\in \cX_r.$ Then in the sum on the right-hand side of
\eqref{eq:le} we have that either $\delta\in \cX_r$ or $\delta+v\in\cX_r,$ and
  $$
    Z_v(W^{-})\le 2Z_v(W)+(q-2)Z_{\max}^{(r)}(W),
  $$
implying \eqref{eq:le0}.
Now take $v\in \cX_{r-j}, j\ge 1.$ The sum on $\delta$ in \eqref{eq:le} contains
$q/2$ terms with $\delta\in\cX_r,$ $q/4$ terms with $\delta\in \cX_{r-1},$
and so on, before reaching $\cX_{r-j}.$ Finally, let $\delta\in \cup_{i=j}^{r-1}\cX_{r-i}\backslash\{-v\}.$
There are $(q/2^{j})-2$ possibilities, and for each of them either $v+\delta$ or
$\delta$ is in $\cX_{r-j}.$ This implies \eqref{eq:le2} and therefore also \eqref{eq:le1}.
\end{proof}

\vspace*{.05in}
In particular, take $j=0.$
Relations \eqref{eq:square}, \eqref{eq:le0} imply that
  \begin{gather}
Z_{\max,n+1}^{(r)} = (Z_{\max,n}^{(r)})^2\text{  if } B_{n+1} = + \label{eq:r-square}\\
 Z_{\max,n+1}^{(r)} \leq qZ_{\max,n}^{(r)} \text{  if } B_{n+1} = -. \label{eq:r-le}
  \end{gather}
Iterated random maps of this kind were studied in \cite{dia99} which contains
general results on their convergence and stationary distributions.
We need more detailed information about this process, established in the following
lemma.

\vspace*{.05in}\begin{lemma} \label{lemma:conv} Let $U_n, n\ge 0$ be a sequence of random variables adapted to a filtration
$\cF_n$ with the following properties:\\
(i) $U_n\in[0,1]$\\
(ii) $P(U_{n+1}=U_n^2|\cF_n)\ge 1/2$\\
(iii) $U_{n+1}\le q U_n$ for some $q\in \integers_+.$\\
Then there are events $\Omega_0,\Omega_1$ such that $P(\Omega_0\cup\Omega_1)=1$
and $U_n(\omega)\to i$ for $\omega\in\Omega_i, i=0,1.$
\end{lemma}
\begin{proof} (a) First let us rescale the process $U_n$ so that in the neighborhood of zero
it has a drift to zero. Let $\beta\in(0,1)$ be such that
   $$
     q^\beta-1 <1/4.
  $$
Let $X_n=U_n^\beta$. Take $\tau(\omega)$ to be the first time when
$X_n(\omega)\ge 1/2.$ Let $Y_n=X_{\min(n,\tau)}.$ On the event $Y_n\ge 1/2$ we have
 $Y_n=Y_{n+1}$ or
    $$
    E(Y_{n+1}-Y_n|\cF_n)=0
$$
while on the event $Y_n<1/2$ we have
   \begin{align*}
  E(Y_{n+1}-Y_n|\cF_n)&\le \frac12(Y_n^2-Y_n)+\frac12(q^\beta Y_n-Y_n)\\
   &\le -\frac18 Y_n\le 0.
  \end{align*}
This implies that the sequence $Y_n, n\ge 0$ forms a supermartingale which is
bounded between $0$ and $1.$
By the convergence theorem, $Y_n\to Y_\infty$ a.e. and in $L^1(\Omega,\cF,P),$
where $Y_\infty$ is a random variable supported on $[0,1].$
This implies that $EY_0\ge EY_n\downarrow EY_{\infty}.$ Further, if
$X_0\in[0,1/4]$ then (since $EY_0=EX_0$)
   \begin{equation}\label{eq:012}
    P(Y_\infty\ge 1/2)\le 2EY_0\le 1/2.
   \end{equation}

\vspace*{.05in}(b) Now we shall prove that $P(Y_\infty\in(\delta,\frac 12 -\delta))=0$ for
any $\delta>0.$  From (ii) it follows that
$P(X_{n+1}=X_n^2|\cF_n)\ge 1/2,$ which implies that
  \begin{equation}\label{eq:1/2}
    P(Y_{n+1}=Y_n^2|\cF_n)\ge 1/2 \quad\text{on }Y_n<1/2
  \end{equation}
for all $n\ge 0.$
Suppose that $Y_\infty$ takes values in $(\delta,1/2-\delta)$ with probability
$\alpha>0.$ Let $A_n=\{\omega: Y_n\in (\delta,1/2-\delta)\}.$ Since $Y_n\to Y_\infty$
a.e., the Egorov theorem implies that there is a subset of probability arbitrarily close
to $P(A_n)$ which this convergence is uniform, and thus $P(A_n)\ge \alpha/2$ for all sufficiently
large $n$. Therefore
  \begin{align*}
  P(|Y_{n+1}-Y_n|\ge \delta^2/2)&\ge P(Y_{n+1}=Y_n^2, Y_n\in(\delta,1/2-\delta))\\
   &\ge \frac \alpha4,
  \end{align*}
the last step by \eqref{eq:1/2}. This however contradicts the almost sure convergence
of $Y_n.$

\vspace*{.05in}
(c) This implies that $P(Y_\infty<1/2)=P(Y_n\to 0)=P(U_n\to 0).$
From \eqref{eq:012}
   \begin{equation}\label{eq:1/4}
  P(U_n\to 0)\ge \frac 12 \quad\text{provided that } U_0\le\Big(\frac14\Big)^{\frac1\beta}.
   \end{equation}
Moreover, if $U_0\le (1/2)^{1/\beta}$ then either $Y_n\to 0$ or $Y_n\ge 1/2$ for some
$n$. This translates to
\begin{equation}\label{eq:or}
P((U_n\to 0) \text{ or } (U_n\ge(1/2)^{1/\beta} \text{ for some $n$}) )=1
\end{equation}
provided that $U_0\le (1/2)^{1/\beta}.$

\vspace*{.05in}
(d) Let $\delta>0$ be such that $q(\frac12)^{\frac1\beta}<1-\delta$ (depending on
$q$ this may require taking a sufficiently small $\beta$).
Let $L:=[0,(\frac14)^{\frac1\beta}]$ and $R:=[1-\delta,1].$
Observe that the process $U_n$ cannot move from $L$ to $R$ without visiting
$C:=((\frac12)^{\frac1\beta},1-\delta).$
Let $\sigma_1$ be the first time when $U_n\in C,$
let $\eta_1$ be the first time after $\sigma_1$ when $U_n\in L\cup R,$
let $\sigma_2$ be the first time after $\eta_1$ when $U_n\in C$, etc.,
$\sigma_1<\eta_1<\sigma_2<\eta_2<\dots.$
We shall prove that every sample path of the process
eventually stays outside $C$, i.e., that for almost all $\omega$ there exists
$k=k(\omega)<\infty$ such that $\sigma_k(\omega)=\infty.$

Assume the contrary, i.e., $\lim_{k\to\infty}P(\sigma_k<\infty)= \alpha>0$
(since $P(\sigma_{k+1}<\infty)<P(\sigma_k<\infty)$, this limit exists.)
We have
  \begin{align}
   P(\exists k:&\;\sigma_k=\infty)\ge \sum_{j=1}^\infty P(\sigma_j\ne\infty;
   U_{\eta_j}\in L;\sigma_{j+1}=\infty)\nonumber\\
   &\ge\alpha\sum_{j=1}^\infty P(
   U_{\eta_j}\in L;\sigma_{j+1}=\infty|\sigma_j\ne\infty). \label{eq:cp}
  \end{align}
Consider the process $U'_n=U_{\sigma_k+n}$ on the event $\sigma_k<\infty$
(with the measure renormalized by $P(\sigma_k<\infty)$). This process
has the same properties (i)-(iii) as $U_n.$
Let $J=\lceil\log_2(\frac1\beta\log_{1-\delta}1/4)\rceil,$
then $x^{2^J}\in L$ for any $x\in C.$
Therefore, $P(U'_J\in L)\ge 2^{-J}$  by
property (ii). Now consider the process $U'_{J+n}$ on the
event $U'_J\in L$.
This process has properties (i)-(iii), so we can use \eqref{eq:1/4} to conclude that
for
   $$
   P(   U_{\eta_k}\in L;\sigma_{k+1}=\infty|\sigma_k\ne\infty)\ge 2^{-(J+1)}
   $$
uniformly in $k.$ But then the sum in \eqref{eq:cp} is equal to
infinity, a contradiction.

(e) The proof is completed by showing that the probability of $U_n$ staying in
$R^c=[0,1]\backslash R$ without converging to zero is zero.
We know that almost all trajectories stay outside $C$, so suppose that the
process starts in $(0,(1/2)^{1/\beta}).$ Then the probability that it enters $L$ in a finite
number of steps is uniformly bounded from below
(this is shown similarly to \eqref{eq:cp}), so the probability that
it does not go to $L$ is zero.
Next assume that the process starts in $L,$ then by \eqref{eq:or}
it either goes to zero or enters $C$ with probability one.
Together with part (d) this implies
that the process that starts in $L$ converges to zero or one with probability one.
\end{proof}

\vspace*{.05in}\begin{lemma}\label{lemma:bhattaorder} Let $V:\cX\to \tilde\cY$
be a channel.
Let $v,v'\in \cX\backslash\{0\}$ be such that $\wt_r(v) \geq
\wt_r(v')$. For any $\delta'>0$ there exists $\delta>0$ such
that $Z_{v'}(V) \geq 1 - \delta'$ whenever $Z_{v}(V) \geq 1 - \delta$.
In particular, we can take $\delta=\delta'q^{-3}.$ 
\end{lemma}
\begin{proof} If $\wt_r(v)=1$ then $v=10\dots0,$ so the statement is trivial.
Let $Z_v(V)\ge 1-\delta,$ where $\wt_r(v)=i\ge 2.$
Then for every pair $x,x'=x+v$ we have
$Z(V_{\{x,x'\}})\ge 1-\epsilon,$ where $\epsilon=q\delta.$
Consider the unit-length vectors $z=(\sqrt {V(y|x)}, y\in\tilde \cY),
z'=(\sqrt {V(y|x')}, y\in\tilde \cY),$ and let $\theta(z,z')$
be the angle between them. We have $\cos (\theta(z,z'))=Z(V_{\{x,x'\}})
\ge 1-\epsilon,$ and so 
$\|z-z'\|^2=2-2\cos(\theta(z,z'))\le 2\epsilon.$

Now take a pair of symbols $x_1,x_2=x_1+v'$ where $v'\in \cX_s, s\le i.$ There exists a 
number $t \in \cX_{r-i+s}$ such that
$v' = tv.$ 
Define $z_1 = (\sqrt
{V(y|x_1)}, y\in\tilde\cY)$ and $z_2 = (\sqrt {V(y|x_2)}, y\in\tilde \cY)$.
Let $w_j=(\sqrt {V(y|x_1+j v)}, y\in\tilde \cY), j=1,\dots,t-1.$ 
From the triangle inequality 
   \begin{align*}
   \|z_1 - z_2\| &\leq \|z_1-w_1\|+\|w_1-w_2\|+\dots+\|w_{t-1}-z_2\|\\
      &\le t\sqrt{2\epsilon} \\
   &\leq
q \sqrt{2\epsilon}.
  \end{align*}
We obtain
  \begin{align*}
    Z(V_{\{x_1,x_2\}})&=\cos(\theta(z_1,z_2))=1-\half\|z_1-z_2\|^2\\
     &\ge 1-q^2\epsilon\\
     &=1-q^3\delta.
  \end{align*}
Thus we obtain
   $$
     Z_{v'}(V)=\frac1q\sum_x Z(V_{\{x,x+v\}})\ge 1-q^3\delta.
   $$
\end{proof}

\vspace*{.05in} {\em Remark :} We can prove the previous lemma in a 
different way by relating the Bhattacharyya distance to the $\ell_1$-distance 
between $V(y|x_1)$ and  $V(y|x_2)$
\cite{par11b}. Then the estimate $\delta=\delta'q^{-3}$ 
can be improved to $\delta=\delta'(2q)^{-2}.$

\vspace*{.05in}\begin{lemma} \label{lemma:convZmax}
For all $j=1,\dots,r$
  $$
   Z_{\max,n}^{(j)} \stackrel{\text{a.e.}}{\longrightarrow}Z_{\max,\infty}^{(j)}.
  $$
where $Z_{\max,\infty}^{(j)}$ is a Bernoulli random variable supported on $\{0,1\}.$
\end{lemma}
\begin{proof} 
For a given channel $V$ denote
  $$
Z_{\max}^{[s,r]}(V)=\max(Z_{\max}^{(s)}(V),Z_{\max}^{(s+1)}(V),\dots,Z_{\max}^{(r)}(V)).
  $$
Eq. \eqref{eq:v-square} gives us that
  $$
  Z_{\max}^{[r-j,r]}(W^+)=(Z_{\max}^{[r-j,r]}(W))^2
  $$
and \eqref{eq:le2} implies that
  $$
  Z_{\max}^{[r-j,r]}(W^-)\le q Z_{\max}^{[r-j,r]}(W).
  $$
Hence by Lemma \ref{lemma:conv} the random variables $Z_{\max,\infty}^{[r-j,r]}$ are
  well-defined and are Bernoulli 0-1 valued a.e. for all $j=0,1,\dots,r-1.$

We need to prove the same for $Z_{\max,\infty}^{(r-j)}.$ The proof is by induction on $j$. We just established the needed claim for $Z^{(r)}_{\max, n}.$
For ease of understanding let us show that this implies the convergence of $Z^{(r-1)}_{\max, n}.$
Indeed, $Z_{\max,\infty}^{[r-1,r]}$ is a Bernoulli 0-1 valued random variable.
 But so is $Z_{\max,\infty}^{(r)},$ so the possibilities are
  $$
   (Z_{\max,\infty}^{[r-1,r]},Z_{\max,\infty}^{(r)})=(1,1) \text{ or } (1,0) \text{ or } (0,0)
  $$
with probability one (note that $(0,1)$ is ruled out by the definition of $Z_{\max}^{[r-1,r]}$).
If $Z_{\max,\infty}^{(r)}=1$ then $Z_{\max,\infty}^{(r-1)}=1$ by Lemma \ref{lemma:bhattaorder}
(this statement holds trajectory-wise).
If on the other hand, the case that is realized is $(1,0)$ then $Z_{\max,\infty}^{(r-1)}=1$
by the definition of $Z_{\max}^{[r-1,r]}.$ Finally in the case $(0,0)$ we clearly have
that $Z_{\max,\infty}^{(r-1)}=0,$ both holding trajectory-wise.

The general induction step is almost exactly the same. Assume that we have proved the required
convergence for $Z_{\max}^{(r-i)}, i=0,1,\dots,j-1.$
Assume that $Z_{\max,\infty}^{[r-j,r]}=0,$ then $Z_{\max}^{(r-j)}=0.$ If on the other hand,
$Z_{\max,\infty}^{[r-j,r]}=1$ then either one of $Z_{\max,\infty}^{(r-i)},i<j$ equals one, and then
 $Z_{\max,\infty}^{(r-j)}=1$ by Lemma \ref{lemma:bhattaorder}, or $Z_{\max,\infty}^{(r-i)}=0$
for all $i<j$, and then $Z_{\max,\infty}^{(r-j)}=1$ by definition of $Z_{\max,\infty}^{[r-j,r]}.$
\end{proof}

\vspace*{.05in}
Now we are in a position to complete the proof of convergence.
\begin{lemma}\label{lemma:Zv} $Z_{v,n}\to Z_{v,\infty}$ a.e., where
$Z_{v,\infty}$ is a $(0,1)$-valued random variable whose distribution depends only
on the ordered weight $\wt_r(v).$
\end{lemma}
\begin{proof}
Let $\Omega_i^{(j)}=\{\omega: Z_{\max,n}^{(j)}\to i\},$ where $i=0,1$ and $j=1,\dots,r,$
where some of the events may be empty. 
For every $\omega\in \Omega_1^{(j)},j=1,\dots,r$ we have that for any $\delta>0$
starting with some $n_0$ the quantity $Z_{\max,n}^{(j)}\ge 1-\delta.$ Thus,
for $n\ge n_0$ there exists $v\in \cX_j,$ possibly depending on $n$,
such that $Z_{v,n}(\omega)\ge 1-\delta.$
Then Lemma \ref{lemma:bhattaorder} implies that $Z_{v',n}(\omega)\ge 1-q^3\delta$
for all $v'\in \cX_j,$ so $Z_{v,n}(\omega)\to 1.$ At the same time, if
$\omega\in\Omega_0^{(j)}$ then $Z_{v,n}(\omega)\to 0$ for all $v\in \cX_j.$
\end{proof}

\vspace*{.1in}
\subsubsection{Proof of Part (b) of Theorem \ref{thm:main}}

\vspace*{.1in}\begin{lemma}\label{lemma:Z} For any $i=1,\dots,r,$ the random variable $Z_{i,n}$
  converges a.e. to a $(0,1)$-valued random variable $Z_{i,\infty}$. Moreover,
$Z_{i,\infty}\ge Z_{i-1,\infty}$ a.e.
\end{lemma}
\begin{proof}
The first part follows because all the $Z_v, v\in \cX_i$
converge to identical copies of the same random variable. Formally, Lemma \ref{lemma:Zv}
asserts that $Z_{v,n}\to j$ for every $v\in \cX_i$ and every $\omega\in \Omega_j^{(i)}, j=0,1.$
Hence taking the limit $n\to\infty$ in \eqref{eq:zi} we see that $Z_{i,n}\to j$ on $\Omega_j^{(i)}$
where $P(\Omega_0^{(i)}\cup\Omega_1^{(i)})=1.$

Let us prove the second part.
Suppose that $Z_{i,n}\ge 1-\epsilon',$ then using \eqref{eq:zi} we see that
$Z_{v',n}\ge 1-2^{i-1}\epsilon'$ for all $v'\in \cX_i.$ Lemma \ref{lemma:bhattaorder}
implies that $Z_{v,n}\ge 1- 2^{3r+i-1}\epsilon'$ for any $v\in\cX, \wt_r(v)=i,$
and therefore $Z_{i,n}\ge 1- 2^{3r+i-1}\epsilon'.$ Thus $Z_{i,n}(\omega)\to 1$ implies
$Z_{i-1}(\omega)\to 1$ for all $\omega\in \Omega_1(i)$ and all $i.$
The second claim of the lemma now follows because $Z_{i,\infty}$ are 0-1 valued
for all $i.$
 \end{proof}
\remove{
Suppose that $\bfv\in\cX, \wt_r(\bfv)=i,$ and that $Z_\bfv(V)\geq 1-\delta$ for
some $\delta>0.$ This implies that $Z(V_{\{\bfx,\bfx+\bfv\}})\ge 1 - \epsilon$ for
all $\bfx$ with $\epsilon = 2^r\delta.$

From Lemma \ref{lemma:bhattaorder}, if $Z(V_{\{\bfx,\bfx+\bfv\}})\ge 1 - \epsilon$,
$Z(V_{\{\bfx,\bfx+\bfv'\}})\ge 1 - 4q\epsilon$ for every $\bfv'$ with $\wt_r(\bfv')
= i$.
Then $Z_{i,n}$ is lower bounded by
\begin{align*}
Z_{i,n} = \frac1{2^{i-1}}
\sum_{\begin{substack}{\bfv\in\cX\\ \wt_r(\bfv)=i}\end{substack}} Z_{\bfv}(V)
\geq 1 - 4q\epsilon.
\end{align*}

{\em Completing the proof of Theorem \ref{thm:main}:}
Now suppose that $Z_{\bfv}(V) \leq \delta$ for some $\delta > 0$. Suppose
that there exists $\bfv'$ with $\wt_r(\bfv') =
\wt_r(\bfv)$
such that $Z_{\bfv'}(V) > \delta$. By convergence, this implies that there
exists a number $n_0$ such that for
Since $Z_{\bfv,n}$ converges to a $(0,1)$-valued random variable, for
a large number $n$, if $Z_{\bfv'}(V) \nleq \delta$ there exists some $\delta' > 0$
such that $Z_{\bfv'}(V) \geq 1- \delta'$. From the previous claim, $Z_{\bfv}(V) \geq
1 - 4q\epsilon$ where $\epsilon = 2^r\delta'$. This is a contradiction to our
assumption. Therefore, every $\bfv'$ with $\wt_r(\bfv')= \wt_r(\bfv)$, $Z_{\bfv'}(V)
\leq \gamma$ for some $\gamma > 0$ and this implies that $Z_{i,n}$ converges to $0$
whenever $Z_{\bfv,n}$ converges to $0$.}

\vspace*{.1in}
We obtain that $Z_{i,\infty}$ is a $(0,1)$ random variable a.e. and for all $i$, and
if $Z_{i,\infty}=1$ then $Z_{j,\infty}=1$ for all $1\le j< i.$
Consider the events $\Psi_i^{(j)}=\{\omega: Z_{j,\infty}=i\}, i=0,1; j=1,\dots,r.$
We have
  \begin{gather*}
   \Psi_1^{(1)}\supset\Psi_1^{(2)}\supset\dots\supset\Psi_1^{(r)}\\
   \Psi_0^{(1)}\subset\Psi_0^{(2)}\subset\dots\subset\Psi_0^{(r)}.
  \end{gather*}
We need to prove that with probability one, the vector $(Z_{i,\infty},i=1,\dots,r)$
takes one of the values \eqref{eq:triangular}. With probability one $Z_{r,\infty}=1$ or $0$.
If it is equal to $1$ then necessarily $Z_{r-1,\infty}=\dots=Z_{1,\infty}=1.$
Otherwise $Z_{r,\infty}=0.$  In this case it is possible that
$Z_{r-1,\infty}=1$ (in which case $Z_{r-2,\infty}=\dots=Z_{1,\infty}=1$)
or $Z_{r-1,\infty}=0.$ Of course $P(\Psi_0^{(r-1)}\cup\Psi_1^{(r-1)})=1,$ so
in particular
  $$
  P(\Psi_0^{(r)}\backslash(\Psi_0^{(r-1)}\cup(\Psi_1^{(r-1)}\backslash \Psi_1^{(r)})))
   =0.
  $$
If $Z_{r-1,\infty}=0$ then the possibilities are $Z_{r-2,\infty}=1$ or $0$, up to
another event of probability 0, and so on.
Thus, the union of the disjoint events given by \eqref{eq:triangular} holds with probability one.
Theorem \ref{thm:main} is proved.  \hfill \QED

\vspace*{.05in}\subsubsection{Proof of Prop.~\ref{prop:finite}}
The proof is analogous to the argument in the previous paragraph.
The random variable $Z_{r,n}\to Z_{r,\infty}$ a.e. .
By the Egorov theorem, for any $\gamma>0$ there are disjoint subsets
$\widetilde\Psi^{(r)}_0\subset \Psi^{(r)}_0,\widetilde\Psi_1^{(r)}\in \Psi^{(r)}_1$ with
$P(\widetilde\Psi^{(r)}_0\cup\widetilde\Psi^{(r)}_1)\ge 1-\gamma$ on which
this convergence is uniform. Take $n_1^{(r)}$ such that $Z_{r,n}>1-\epsilon/2^{4r-1}$
for every $\omega\in\widetilde\Psi_1^{(r)}$ and $n\ge n_1^{(r)}.$
By Lemma \ref{lemma:bhattaorder} and \eqref{eq:zi}
for every such $\omega$ we have $Z_{i,n}\ge 1-\epsilon$ for all $i=1,\dots,r-1;$
$n\ge n_1^{(r)}.$
This gives rise to the event $B_{r,n}.$ Otherwise let $n_0^{(r)}$ be such that
$\sup_{\omega}Z_{r,n}<\epsilon$ for $\omega\in \widetilde\Psi^{(r)}_0$ and $n\ge n_0^{(r)}.$
Consider the events $\widetilde\Psi^{(r-1)}_0\subset \Psi^{(r-1)}_0,\widetilde
\Psi_1^{(r-1)}\subset \Psi^{(r-1)}_1$ with
$P(\widetilde\Psi^{(r-1)}_0\cup\widetilde\Psi^{(r-1)}_1)\ge 1-\gamma$
on which $Z_{r-1,n}\to Z_{r-1,\infty}$ uniformly.
Choose $n_1^{(r-1)}$ such that $Z_{r-1,n}>1-\epsilon/2^{4r-2}$
for all $n\ge n_1^{(r-1)}$ and all $\omega \in \widetilde\Psi_1^{(r-1)}.$
For every such $\omega$ we have $Z_{i,n}\ge 1-\epsilon$ for all $i=1,\dots,r-2;$
$n\ge n_1^{(r-1)}.$ Next,
  $$
  P(\widetilde\Psi_0^{(r)}\backslash(\widetilde\Psi_0^{(r-1)}
\cup(\widetilde\Psi_1^{(r-1)}\backslash \widetilde\Psi_1^{(r)})))
   \le 2\gamma.
  $$
We continue in this manner until we construct all the $r+1$ events $B_{k,n}.$
For this, $n$ should be taken sufficiently large, $n\ge\max_{k}\max(n_0^{(k)},n_1^{(k)}).$
By taking $\gamma=\delta/r$ we can ensure that  $P(\cup_k B_{k,n}\ge1-\delta.$
This concludes the proof.

\vspace*{.05in} {\em Remark :} For binary-input channels, the transmitted bits in the limit are
transmitted either perfectly or carry no information about the message.
{\c S}a{\c s}o{\u g}lu et al. \cite{sas09} observed that $q$-ary
codes constructed using Ar\i kan's kernel $H_2$ share this property for transmitted
symbols only if $q$ is prime. Otherwise \cite{sas09} notes the symbols can polarize to states that
carry partial information about the transmission. In particular, they give an
example of a quaternary-input channel $W:\{0,1,2,3\}\to\{0,1\}$ with
$W(0|0)=W(0|2)=W(1|1)=W(1|3)=1.$ This channel has capacity 1 bit. Computing
the channels $W^+$ and $W^-$ we find that they are equivalent to the original
channel $W$. The conclusion reached in \cite{sas09} was that
there are nonbinary channels that do not polarize under the action of $H_2.$

We observe that the above channel corresponds to the extremal configuration $10$
in \eqref{eq:triangular} (the other two configurations arise with probability 0),
and therefore has to be, and is, a stable point of the channel combining operation.
It is possible to reach capacity by transmitting the least significant bit of every symbol.

Paper \cite{sas09} went on to show that for every $n\ge 1$ there exists a permutation $\pi_n:
\cX\to \cX$ such that the kernels $H_2(n):(u,v)\to (u+v, \pi_n(v))$ lead to channels
that polarize to perfect or fully noisy. While the result of \cite{sas09} holds for any
$q$, in the case of $q=2^r$
this means that configurations $00\dots 0$ and $11\dots 1$ arise
with probability $1-I(W)$ and $I(W)$ respectively, while all the other configurations
have probability zero.

\subsection{Rate of polarization and error probability of decoding}\label{sect:error}
The following theorem, due to Ar\i kan and Telatar \cite{ari09a}, is useful in
quantifying the rate of convergence of the channels $W_n$ to one of the extremal
configurations \eqref{eq:triangular}.
\begin{theorem} \label{thm:rate} \cite{ari09a} {\sl Suppose that a random process
$U_n,n\ge 0$ satisfies the
conditions
(i)-(iii) of Lemma \ref{lemma:conv} and that (iv), $U_n$ converges a.e. to a $\{0,1\}$-valued
random variable $U_\infty$ with $P(U_\infty=0)=p$. Then for any $\alpha\in(0,1/2)$
   \begin{equation}\label{eq:rate}
   \lim_{n\to\infty} P(U_n<2^{-2^{\alpha n}})=p.
   \end{equation}
If condition (iii) is replaced with (iii$'$) $U_n\le U_{n+1}$ and $U_0>0$, then
for any $\alpha>1/2,$
  $$
   \lim_{n\to\infty} P(U_n < 2^{-2^{\alpha n}})=0.
   $$}
\end{theorem}
Note that, as a consequence of Lemma \ref{lemma:conv}, assumption (iv) in this
theorem is superfluous in that it follows from (i)-(iii).

Processes $Z_{\max,n}^{(r)}$ and $Z_{\max,n}^{[r-j,r]}, j=0,\dots,r-1$ satisfy conditions
(i)-(iii) of Lemma \ref{lemma:conv}. Hence the above theorem gives the rate of convergence
of each of them to zero. We argue that the convergence rate of $Z_{\max,n}^{(r-j)},j\ge 1$
to zero is also governed by Theorem \ref{thm:rate}. Indeed, let
$\Omega_i^{[r-j,r]}=\{\omega:
Z_{\max,n}^{[r-j,r]}\to i\}, \Omega_i^{(r-j)}=\{\omega:Z_{\max,n}^{(r-j)}\to i\},i=0,1.$ Then
  \begin{equation}\label{eq:inc}
   \Omega_0^{(r-j)}\supseteq \Omega_0^{[r-j,r]} \text{ and }
    \Omega_1^{(r-j)}=\Omega_1^{[r-j,r]}
  \end{equation}
the last equality
because by Lemma \ref{lemma:bhattaorder}, $Z_{\max,n}^{[r-j,r]}\to 1$ implies $Z_{\max,n}^{(r-j)}\to 1$ on every trajectory. As a consequence of \eqref{eq:inc} we have that
  $
  P(\Omega_0^{(r-j)}\backslash \Omega_0^{[r-j,r]})=0.
  $
Hence $P(Z_{\max,\infty}^{(r-j)}=0)=P(Z_{\max,\infty}^{[r-j,r]}=0)$.
Denote this common value by $p_j.$
The random variable $Z_{\max,n}^{[r-j,r]}$ satisfies a condition of the form \eqref{eq:rate}
with $p=p_j.$
We obtain that for any $\alpha\in(0,1/2)$
  $$
  \lim_{n\to\infty}P(Z_{\max,n}^{(r-j)}<2^{-2^{\alpha n}})
=\lim_{n\to\infty}P(Z_{\max,n}^{(r-j)}<2^{-2^{\alpha n}})  =p_j.
  $$
Of course if $Z_{\max,n}^{(r-j)}$ is small then so is every $Z_{v,n}$ for $v\in \cX_{r-j}$.
We conclude as follows.
\begin{proposition} For any $\alpha\in(0,1/2)$ and any $v\in \cX_j,j=1,2,\dots,r$
  $$
   \lim_{n\to\infty}P(Z_{v,n}< 2^{-2^{\alpha n}})=p_j.
  $$
\end{proposition}

This result enables us to estimate the probability of decoding error under successive
cancellation decoding.
To do this, we extend the argument of \cite{ari09} to nonbinary alphabets.

The following statement follows directly from the previously established results, notably
Proposition \ref{prop:Ik}.
\vspace*{.05in}
\begin{theorem}\label{thm:coderate} Let $0<\alpha< 1/2.$ For any DMC $W:\cX\to\cY$ with
$I(W)>0$ and any $R<I(W)$ there exists a sequence of $r$-tuples of disjoint subsets
$\cA_{0,N},\dots,\cA_{r-1,N}$ of $[N]$ such that $\sum_k|\cA_{k,N}|(r-k)\ge NR$ and
$Z_{v}(W_N^{(i)})< 2^{-N^\alpha}$ for all $i \in \cA_{k,N},$ all
$v \in \bigcup_{l=k+1}^r\cX_{l},$ and all $k=0,1,\dots,r-1.$
\end{theorem}

\vspace*{.05in} Let
\begin{align*}
    \cE &\triangleq \{(u_1^N,y_1^N)\in\cX^N\times\cY^N : \hat u_1^N \neq u_1^N\} \\
    \cB_i &\triangleq \{(u_1^N,y_1^N)\in\cX^N\times\cY^N : \hat u_1^{i-1}= u_1^{i-1},
\hat u_i \neq u_i\}.
\end{align*}
Then the block error probability of decoding is defined as
  $$
P_e = P(\cE) = P\big(\bigcup_{i\in\cA_{0,N}\cup\dots\cup\cA_{r-1,N}}\cB_i\big).
  $$

\vspace*{.05in} The next theorem is the main result of this section.
\begin{theorem}\label{thm:decodeerror} Let $0<\alpha< 1/2$ and let $0< R < I(W),$
where $W:\cX\to\cY$ is a DMC. The best achievable error probability of block
error under successive cancellation decoding at block length $N=2^n$ and rate $R$ satisfies
  $$
P_e = O(2^{-N^\alpha}).
  $$
\end{theorem}

\begin{proof}
\remove{Let us consider
\begin{align*}
    \cE_i \triangleq \{(u_1^N,y_1^N) \in \cX^N \times \cY^N: \hat u_i \neq u_i\}
\end{align*}
for $i \in \cA_{k,N}$, $k=1,\dots,r$. Then $\cB_i \subset \cE_i$. Therefore, we have
\begin{align*}
    P(\cE) \leq \sum_{i\in\cA_N} P(\cE_i).
\end{align*}
To get an upper bound on $P(\cE_i)$, consider}
Let
\begin{align*}
    \cE_{i,v} \triangleq \{&(u_1^N,y_1^N) \in \cX^N \times \cY^N:
    \\ &W_N^{(i)}(y_1^N,u_1^{i-1}|u_i) \leq
    W_N^{(i)}(y_1^N,u_1^{i-1}|u_i+v)\}.
\end{align*}
For a fixed value of $a_1^k=(a_1,a_2,\dots,a_{k})\in\{0,1\}^k$ let us define
$
    \cX(a_1^k) = \{x\in\cX:x_1^k=a_1^k\}.
$
Notice that the decoder finds $\hat u_i$, $i \in \cA_{k,N}$ by taking
the maximum over the symbols $x\in \cX(a_1^k)$. Then we obtain
\begin{align*}
    \cB_i \subseteq \bigcup_{v\in\cX(a_1^k)} \cE_{i,v}.
\end{align*}
Using \eqref{eq:WNi}, we obtain
\begin{align*}
    &P(\cB_i)\leq \sum_{v\in\cX(a_1^k)}P(\cE_{i,v})\\
    &=\sum_{v\in\cX(a_1^k)} \sum_{u_1^N,y_1^N}
    \frac{1}{q^N}W_N(y_1^N|u_1^N)1_{\cE_{i,v}}(u_1^N,y_1^N)\\
    &\leq\sum_{v\in\cX(a_1^k)} \sum_{u_1^N,y_1^N}
    \frac{1}{q^N} W_N(y_1^N|u_1^N)
    \sqrt{\frac{W_N^{(i)}(y_1^N,u_1^{i-1}|u_i+v)}{W_N^{(i)}(y_1^N,u_1^{i-1}|u_i)}}\\
    &=\sum_{v\in\cX(a_1^k)}
    \sum_{u_i}\frac{1}{q}Z(W_{N,\{u_i,u_i+v\}}^{(i)})\\
    &=\sum_{v\in\cX(a_1^k)} Z_v(W_N^{(i)}).
\end{align*}
Thus the decoding error is bounded by
\begin{align*}
    P(\cE) \leq \sum_{i\in \cA_{0,N}\cup\dots\cup\cA_{r-1,N}}\sum_{v\in\cX(a_1^k)} Z_v(W_N^{(i)}).
\end{align*}
By Theorem \ref{thm:coderate}, for any $R < I(W)$ there exists a sequence of $r$-tuples of
disjoint subsets $\cA_{0,N},\dots,\cA_{r-1,N}$ with $\sum_{k}|\cA_{k,N}|(r-k)\geq NR$ such that
\begin{align*}
    \sum_{i\in \cA_{0,N}\cup\dots\cup\cA_{r-1,N}}\sum_{v\in\cX(a_1^k)} Z_v(W_N^{(i)})
    \leq qN2^{-N^\alpha}
\end{align*}
and thus we obtain that $P(\cE) = O(2^{-N^\alpha})$.
\end{proof}

\section{Ordered Channels}\label{sect:ordered}

To compute a few examples, consider ``ordered symmetric channels,'' called so because
they provide a natural counterpart to the combinatorial definition of the
ordered distance \cite{par11a}. A simple example is given by the ordered erasure channel,
defined as
$W_r:\ff_q^r \rightarrow
(\ff_q\cup\{?\})^r,$ where
   $$
W_r(y|x)=\begin{cases}
        \epsilon_0, & y=x, \\
      \epsilon_i, & y=(??\dots?x_{i+1}\dots x_r), 1\le i\le r
                     \end{cases}
  $$
and $W_r(y|x)=0$ if $y$ does not contain any erased coordinates and $y \neq x$.
Its capacity equals $r-\sum_{i=1}^r i\epsilon_i$ and is attained by sending $r$
independent streams of data encoded for binary erasure channels with erasure
probabilities $\sum_{j=i}^r \epsilon_{j}, i=1,\dots,r.$ Therefore, sending $r$
independent polar codewords over the $r$ bit channels, one can approach the capacity
of the channel.

Despite the fact that this example is trivial, it already shows the domination
pattern observed in Theorem \ref{thm:main}. Namely, it is easy to prove directly that
$Z_{j,\infty}\ge Z_{i,\infty}$ a.s. for all $i>j,$ thereby establishing the result of
Lemma \ref{lemma:Z}. For that it suffices to observe that the erasure in
higher-numbered bits implies that all the lower-numbered bits are erased with
probability 1.  We include two examples. In Fig.~\ref{fig1}, $r=2,$
and $\epsilon_0=0.5, \epsilon_1=0.4, \epsilon_2=0.1$. In Fig.~\ref{fig2}, $r=9$ and
$\epsilon_i=0.1, i=0,1,\dots,9.$ Note that the proportion of the channels with
capacity $i=0,1,\dots,r$ bits converges to $\epsilon_i.$

Another example is given by the {\em ordered symmetric channel}  \cite{par11a} which is a DMC $W:\{0,1\}^r\to\{0,1\}^r$
defined by the matrix $W(\bfy|\bfx)$ where
   \begin{equation}\label{eq:OSC}
   W(y|x)=2^{-(j-1)} \epsilon_j
   \end{equation}
for all pairs $y,x$ such that $d_r(x,y)=j , \quad j=1,\dots, r,$
and where $W(x|x)=\epsilon_0$ for all $x\in \cX$.
The ordered symmetric channel models transmission over $r$ parallel links such that,
if in a given time slot a bit is received incorrectly, the bits with indices lower than
it are equiprobable. This system was proposed in \cite{tav06} as an abstraction of transmission
in wireless fading environment.
The capacity of the channel equals
  $$
  I(W)=r+\epsilon_0\log_q\epsilon_0+\sum_{i=1}^r{\epsilon_i}\log_q\Big(\frac{\epsilon_i}{q^{i-1}(q-1)}
  \Big).
  $$
By Theorem \ref{thm:main}
$q$-ary polar codes, $q=2^r$ can be used to transmit at rates close to capacity on this channel;
moreover, the domination pattern that emerges, exactly matches the fading nature of the bundle of
$r$ parallel channels, achieving the capacity of the system discussed above. \remove{We note
an earlier use of polar codes to parallel channels \cite{hof10}; however the
transmission model and the results in that work are very different from our approach.
In particular, \cite{hof10} makes no mention of multilevel polarization.}

\begin{figure}
\centering
\includegraphics[width=3.2in]{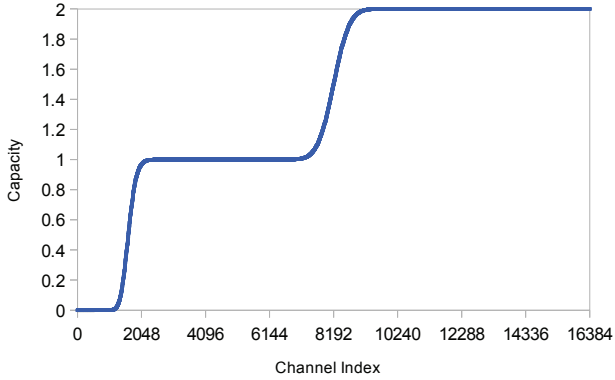}
\caption{ 3-level polarization on the ordered erasure
  channel $W:\cX\to\cY, \cX=\{00,01,10,11\}$ with
  transition probabilities $\epsilon_0:=W(00|00)=0.5, \epsilon_1:=W(?x_2|x_1x_2)=0.4,
  \epsilon_2:=W(??|x_1,x_2)=0.1$, for all $x_1,x_2\in\{0,1\}$. In this example it is easy to see
that $P(I_\infty=i)=\epsilon_i,i=0,1,2.$}\label{fig1}
\end{figure}

\begin{figure}
\centering
\includegraphics[width=3.2in]{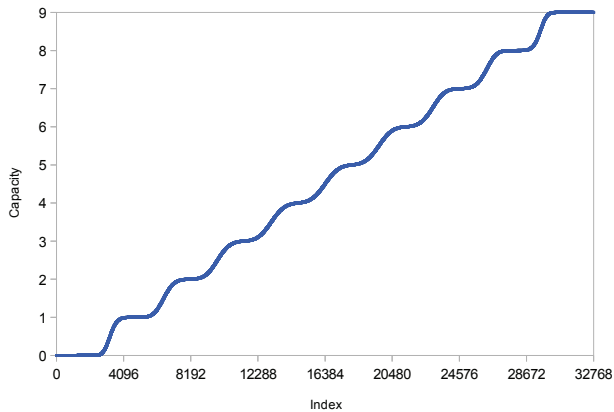}
\caption{10-level polarization on the ordered erasure
  channel $W:\{0,1\}^9\to \cY$ with transition probabilities $\epsilon_i=0.1,i=0,1,\dots,9.$}\label{fig2}
\end{figure}

\section{Conclusion}
The result of this paper offers more detailed information about polarization
on $q$-ary channels, $q=2^r.$ The multilevel polarization adds flexibility
to the design of the transmission scheme in that we can adjust the number of
symbols that carry a given number of bits to a specified proportion of the
overall transmission as long as the total number of bits is fixed.
This could be useful in the design of signal constellations for coded modulation,
including BICM \cite{Alvarado2010,cal93} as well as in other communication problems
that can benefit from nonuniform symbol sets.

The authors are grateful to Emmanuel Abbe, Eren {\c S}a{\c s}o{\u g}lu,
and Emre Telatar (EPFL), and Leonid Koralov, Armand Makowski,
and Himanshu Tyagi (UMD) for useful discussions of this work.
This research was partially supported by NSF grants CCF0916919, CCF0830699,
and DMS1117852.

\section*{Appendix}

\emph{The proof of \eqref{eq:upper}} : We shall
break the expression for $I(W)$ into a sum of symmetric capacities of B-DMCs.

Let $z=(z_1,\dots,z_k)$ be an $k$-tuple of symbols from
$\cX$. Define the probability distribution $P(y|z)=\frac
1k\sum_{i=1}^k W(y|z_i).$ Define a B-DMC
$W^{(k)}_{\{z^{(1)},z^{(2)}\}}:\cX^k\to \cY$ with inputs
$z^{(i)}\in \cX^k$, where the transition $z^{(i)}\to y$ is
given by $P(y|z^{(i)}),$ $i=1,2$.

\begin{lemma}\label{lemma:mixedbhatta}
The Bhattacharyya parameter of the channel $W^{(k)}_{\{z^{(1)},z^{(2)}\}},$
where $z^{(1)}=(x_1,\dots,x_k),z^{(2)}=(x_{k+1},\dots,x_{2k}),$
 can be
lower bounded by
\begin{align}
Z(W^{(k)}_{\{z^{(1)},z^{(2)}\}}) \geq \frac{1}{k} \sum_{j=1}^{k}
Z(W_{\{x_j,x_{f(j)}\}})\label{eq:mixedbhattalower}
\end{align}
for any $f$ which is a one-to-one mapping from the set
$\{1,2,\dots,k\}$ to $\{k+1,\dots,2k\}$.
\end{lemma}

\begin{proof} It suffices to prove the
above inequality for some one-to-one mapping.
Let $f(i) = k + i.$ For brevity denote $w_{i,y}=W(y|x_i). $
We have
\begin{align*}
Z(W^{(k)}_{\{z^{(1)},z^{(2)}\}}) = \frac{1}{k}\sum_{y}\sqrt{\bigg(\sum_{i=1}^k
w_{i,y}\bigg) \bigg(\sum_{i=k+1}^{2k} w_{i,y}\bigg)},
\end{align*}
while the right hand side of (\ref{eq:mixedbhattalower}) is
\begin{align*}
\frac{1}{k} \sum_{j=1}^{k} Z(W_{\{x_j, x_{f(j)}\}}) =
\frac{1}{k}\sum_{y}\sum_{i=1}^{k}\sqrt{w_{i,y}w_{k+i,y}}.
\end{align*}
The Cauchy-Schwartz inequality gives us
  $$
\Big(\sum_{i=1}^k
w_{i,y}\Big) \Big(\sum_{i=k+1}^{2k} w_{i,y}\Big) \ge
\Big(\sum_{i=1}^{k}\sqrt{w_{i,y}w_{k+i,y}}\Big)^2
  $$
hence the lemma.
\end{proof}

Let us introduce some notation.
Given $z=(z_1,\dots,z_k)\in\cX^k,$ let $z \oplus x =
(z_1\oplus x, \dots, z_k\oplus x)$ where $\oplus$ is a
bit-wise modulo-$2$ summation. In the next lemma we consider B-DMCs
$W^{(k)}_{\{z_m^{(1)},z_m^{(2)}\}}:\cX^{k}\to\cY,k=2^{m-1}, m=1,\dots,r$
with inputs of special form. Namely,
\remove{\begin{gather*}
z_1^{(1)}=x_1\\
z_2^{(1)}=(x_1,x_1\oplus x_2)\\
z_3^{(1)}=(x_1,x_1\oplus x_2,x_1\oplus x_3,
x_1\oplus x_2\oplus x_3)\\
\dots
\end{gather*}}
$z_1^{(1)}=x_1;\,
z_2^{(1)}=(x_1,x_1\oplus x_2);\,
z_3^{(1)}=(x_1,x_1\oplus x_2,x_1\oplus x_3,
x_1\oplus x_2\oplus x_3),$ and
generally, $z_m^{(1)}$ is formed of $x_1$ plus all the possible
sums of the vectors $x_2,\dots,x_m$ with $0-1$ coefficients, including the empty one.
Finally, $z_m^{(2)}=z_m^{(1)}\oplus x_{m+1}.$

For $m=0,1,\dots,r-1$ introduce the set
$\cA=\cA(x_1,\dots,x_{m+1})\subset\cX^{m+1}$ as follows:
  \begin{align*}
&\cA=\Big\{(x_1,\dots,x_{m+1})\in \cX^{m+1}\big|
x_1\in\cX;
x_2\in\cX\backslash\{0\};\\
&x_j\ne \sum_{i=2}^{j-1} a_i x_i, \text{ for all choices of } a_i\in\{0,1\}
,
j=3,\dots,m+1\Big\}
\end{align*}

We need the following technical lemma.
\begin{lemma}\label{lemma:sum}
  \begin{equation}\label{eq:sum}
   I(W)=\sum_{m=1}^{r} \bigg(\frac{1}{2^r} \prod_{j=1}^{m} \frac{1}{2^r -
2^{j-1}}\bigg)
\sum_{\cA(x_1,\dots,x_{m+1})}I(W_{\{z_m^{(1)},z_m^{(2)}\}}^{(k)})
 \end{equation}
where the number $k$, the vectors $z_m^{(1)},z_m^{(2)}$, and the set
$\cA(x_1,\dots,x_{m+1})$ are defined before the lemma.
\remove{\begin{align}
z_m^{(1)}&=(x_1,x_1\oplus x_2,x_1\oplus x_3,x_1\oplus x_2\oplus x_3,\dots,\nonumber\\
&\qquadx_1\oplus x_m,\dots,x_1\oplus\dots\oplus x_m),\label{eq:zm1}\\
z_m^{(2)}&=z_m^{(1)}\oplus x_{m+1}.\label{eq:zm2}
\end{align}
while the condition $c(x_2,\dots,x_{m+1})$ of the second
summation means that $x_j \neq 0$ for $j=2,\dots,m+1$ and $x_j$
is not equal to any summation of $x_2,\dots,x_{j-1}$ for $j=3,\dots,m+1.$}
\end{lemma}
\begin{proof} First we express the capacity of $W$ as the sum of
symmetric capacities of B-DMCs.
\begin{align*}
&I(W)\\
&= \frac{1}{2^r}\sum_x\sum_y W(y|x) \log{\frac{W(y|x)}{P(y)}}\\
&=\frac{1}{2^r}\sum_y
\frac{1}{2(2^r-1)}\sum_{x_1}\sum_{x_2:x_2\neq
0} \bigg(W(y|x_1)\log{\frac{W(y|x_1)}{P(y)}}\\
&\qquad\qquad\qquad + W(y|x_1 \oplus
x_2)\log{\frac{W(y|x_1 \oplus x_2)}{P(y)}}\bigg)\\
&=\frac{1}{2^r(2^r-1)}\\
&~\cdot\sum_y\sum_{\begin{substack}{x_1,x_2\\x_2 \neq
0}\end{substack}} \bigg(\frac{1}{2}
W(y|x_1)\log{\frac{W(y|x_1)}{\frac12(W(y|x_1) +
W(y|x_1 \oplus x_2))}}\\
&\qquad+ \frac{1}{2} W(y|x_1 \oplus
x_2)\log{\frac{W(y|x_1 \oplus
x_2)}{\frac12(W(y|x_1) + W(y|x_1 \oplus x_2))}}\\
&\qquad + \frac12(W(y|x_1) + W(y|x_1 \oplus x_2))\\
&\qquad~ \cdot\log{\frac{\frac12(W(y|x_1) + W(y|x_1 \oplus
x_2))}{P(y)}}\bigg)\\
&=\frac{1}{2^r(2^r-1)}\Big\{\sum_{\begin{substack}{x_1,x_2\\x_2 \neq
0}\end{substack}} I(W_{\{x_1, x_1 \oplus x_2\}})+T_2\Big\}
\end{align*}
where
\begin{align*}
&T_2=\sum_y\sum_{\begin{substack}{x_1,x_2\\x_2 \neq
0}\end{substack}}\frac12(W(y|x_1) + W(y|x_1 \oplus x_2))\\
&\qquad\qquad \cdot\log{\frac{\frac12(W(y|x_1) + W(y|x_1 \oplus
x_2))}{P(y)}}\Big\}.
\end{align*}
Observe that the condition $x_2\ne 0$ is needed in order to obtain the
expression for $I(W_{\{x_1, x_1 \oplus x_2\}})$.

We will apply the same technique repeatedly. In the next step we add another
sum, this time on $x_3$ which has to satisfy the conditions $x_3\ne 0,
x_3\ne x_2.$ We have
\begin{align*}\label{eq:T3}
&T_2=
\sum_{y}\frac{1}{2(2^r-2)}\sum_{\cA(x_1,x_2,x_3)}
\bigg(\frac12(W(y|x_1) + W(y|x_1 \oplus
x_2))\nonumber \\
&\qquad~~\cdot\log{\frac{\frac12(W(y|x_1) + W(y|x_1 \oplus
x_2))}{P(y)}}\nonumber\\
&\qquad+ \frac12(W(y|x_1 \oplus x_3) + W(y|x_1 \oplus
x_2 \oplus x_3))\nonumber\\
&\qquad~~\cdot\log{\frac{\frac12(W(y|x_1 \oplus x_3) +
W(y|x_1 \oplus x_2 \oplus x_3))}{P(y)}}\bigg)\nonumber\\
&=
\frac{1}{2^r-2}\sum_y\sum_{\cA(x_1,x_2,x_3)}
\bigg(\frac{1}{2}\cdot \frac12( W(y|x_1) +
W(y|x_1 \oplus x_2)) \nonumber\\
&\qquad~\cdot\log\frac{\frac12(W(y|x_1) + W(y|x_1 \oplus
x_2))}{B} + B\log{\frac{B}{P(y)}}\nonumber\end{align*}
\begin{align*}
&\qquad + \frac{1}{2}\cdot \frac12( W(y|x_1 \oplus x_3) +
W(y|x_1 \oplus x_2 \oplus x_3))\nonumber\\
&\qquad~\cdot\log
\frac{\frac12(W(y|x_1 \oplus x_3) +
W(y|x_1 \oplus x_2 \oplus x_3))}{B}\bigg)\nonumber
\end{align*}
where
$B=\frac14( W(y|x_1) + W(y|x_1 \oplus x_2) +
W(y|x_1 \oplus x_3) +
W(y|x_1 \oplus x_2 \oplus x_3) ).$

By now it is clear what we want to accomplish.
Let us again take the sum on $y$ inside. Recalling the definition
of the channel $W^{(k)}$ before Lemma \ref{lemma:mixedbhatta}, we obtain
\begin{align*}
  T_2=\frac 1{2^r-2}\biggl\{\sum_{\cA(x_1,x_2,x_3)} I(W_{\{z_2^{(1)},z_2^{(2)}\}}^{(2)})
+T_3\biggr\};
\end{align*}
here $I(W_{\{z_2^{(1)},z_2^{(2)}\}}^{(2)})$ is the symmetric capacity of the B-DMC
$W_{\{z_2^{(1)},z_2^{(2)}\}}^{(2)}$ with $z_2^{(1)} = \{x_1,
x_1 \oplus x_2\}$ and $z_2^{(2)} = \{x_1 \oplus x_3,
x_1 \oplus x_2 \oplus x_3\},$ and $T_3$ is the term remaining
in the expression for $T_2$ upon isolating this capacity:
  $$
   T_3=\sum_y\sum_{\cA(x_1,x_2,x_3)}B\log\frac B{P(y)}.
  $$
Now repeat the above trick for $T_3,$ namely, average over all the linear combinations
that this time include the vector $x_4$ and isolate the symmetric capacity
of the channel $W^{(k)}$ that arises. Proceeding in this manner, we obtain
\begin{align*}
I(W)&=\frac{1}{2^r(2^r-1)}\sum_{\begin{substack}{x_1,x_2\\
x_2
\neq 0}\end{substack}} I(W_{\{x1, x_1 \oplus x_2\}})\\
& + \frac{1}{2^r(2^r-1)(2^r-2)}\sum_{\cA(x_1,x_2,x_3)} I(W_{\{z_2^{(1)},z_2^{(2)}\}}^{(2)})\\
& + \frac{1}{2^r(2^r-1)(2^r-2)}\sum_{y}\sum_{\cA(x_1,x_2,x_3)} B\log\frac{B}{P(y)}\\
&= \dots \\
&=\sum_{m=1}^{r} \bigg(\frac{1}{2^r} \prod_{j=1}^{m} \frac{1}{2^r -
2^{j-1}}\bigg)\sum_{\cA(x_1,\dots,x_{m+1})}
I(W_{\{z_m^{(1)},z_m^{(2)}\}}^{(k)})
\end{align*}
where the notation $z_m^{(1)},z_m^{(2)}, \cA(x_1,\dots,x_{m+1})$
is introduced before the statement of lemma.
\end{proof}

We continue with the proof of inequality \eqref{eq:upper}.
The term with $m=1$ in \eqref{eq:sum} equals
\begin{align*}
&\frac{1}{2^r(2^r-1)}\sum_{\begin{substack}{x_1,x_2\\x_2
\neq 0}\end{substack}} I(W_{\{x_1, x_1 \oplus
x_2\}})\\
&\leq\frac{1}{2^r(2^r-1)}
\sum_{\begin{substack}{x_1,x_2\\x_2 \neq
0}\end{substack}} \sqrt{1-Z(W_{\{x_1, x_1 \oplus x_2\}})^2}
\\&=\frac{1}{2^r(2^r-1)} \sum_{d=1}^{r}
\sum_{\begin{substack}{x_1,x_2\\\wt_r(x_2) =
d}\end{substack}}
\sqrt{1-Z(W_{\{x_1, x_1 \oplus x_2\}})^2}\end{align*}
\begin{align*}
&\leq\frac{1}{2^r(2^r-1)} \sum_{d=1}^{r} 2^{r+d-1}\\
&~~\cdot\sqrt{1-\bigg(\frac{1}{2^{r+d-1}}
\sum_{\begin{substack}{x_1,x_2\\\wt_r(x_2)
= d}\end{substack}} Z(W_{\{x_1, x_1 \oplus x_2\}})\bigg)^2}\\
&=\frac{1}{2^r-1} \sum_{d=1}^{r} 2^{d-1} \sqrt{1-Z_d^2}
\end{align*}
where the first inequality is from the relation between the
symmetric capacity and the Bhattacharyya parameter of B-DMCs \cite{ari09},
 and the second inequality follows from the fact that the function
$\sqrt{1-x^2}$ is concave for $0 \leq x \leq 1$.

The terms with $m\ge 2$ in \eqref{eq:sum} will be estimated using
Lemma \ref{lemma:mixedbhatta}. We will choose the map $f$ so that
the $r$-vector
  $$
a(f)=(z^{(1)})_s\oplus (z^{(2)})_{f(s)}
  $$
does not depend on $s.$ For instance, one such map is given in Lemma
\ref{lemma:mixedbhatta}. Moreover, out of all such mappings we take
the one for which $\wt_r(a(f))$ is the smallest.\remove{
For any two
vectors $u,v \in \cX$, we define the inequality $u < v$
according to the left-to-right lexicographic order. Clearly,
$\wt_r(u) \leq \wt_r(v)$.} Then the second term becomes
\begin{align*}
&\frac{1}{2^r(2^r-1)(2^r-2)}\sum_{\cA(x_1,x_2,x_3)}
I(W_{\{z_2^{(1)},z_2^{(2)}\}}^{(2)})\\
&\leq \frac{1}{2^r(2^r-1)(2^r-2)} \sum_{\cA(x_1,x_2,x_3)}
\sqrt{1-Z(W_{\{z_2^{(1)},z_2^{(2)}\}}^{(2)})^2}\\
&\leq \frac{1}{2^r(2^r-1)(2^r-2)} \sum_{\cA(x_1,x_2,x_3)}
\sqrt{1-\frac{D^2}4}\\
&= \frac{1}{2^r(2^r-1)(2^r-2)} \sum_{d=1}^{r} \sum_{\begin{substack}
{\cA(x_1,x_2,x_3)\\
\wt_r(x_3)=d}\end{substack}} \sqrt{1-\frac{D^2}4}\\
&\leq\frac{1}{2^r(2^r-1)(2^r-2)} \sum_{d=1}^{r} 2^r\cdot\alpha_d\\
&~\cdot\sqrt{1-\Bigg(\frac{1}{2^{r+1}\cdot\alpha_d}
\sum_{\begin{substack} {\cA(x_1,x_2,x_3)\\
\wt_r(x_3)=d}\end{substack}}
D\Bigg)^2}\\
&\leq\frac{1}{(2^r-1)(2^r-2)} \sum_{d=1}^{r} \alpha_d \sqrt{1-Z_d^2}
\end{align*}
where
\begin{align*}
D&=Z(W_{\{x_1,x_1 \oplus x_3\}}) +
Z(W_{\{x_1 \oplus x_2, x_1 \oplus x_2 \oplus x_3\}})\\
\alpha_d &= 2^{d-1}\cdot(2^{r+1}-3\cdot2^{d-1}-1)
\end{align*}
which is the number of terms with $\wt_r(x_3) = d, x_1=0$ under the given
condition. Repeating this process, we obtain the claimed result.
The full calculation is cumbersome, but its essence is
captured in the example for $r=3$ which we write out in full:
\begin{align*}
I&(W)= \sum_{m=1}^{3} \bigg(\frac{1}{8} \prod_{j=1}^{m} \frac{1}{8 -
2^{j-1}}\bigg) \sum_{\cA(x_1,\dots,x_{m+1})} I(W_{\{z_m^{(1)},z_m^{(2)}\}}^{(m)})
\end{align*}
\begin{align*}
&=\frac{1}{8\cdot 7}\sum_{\cA(x_1,x_2)} I(W_{\{x_1, x_1 \oplus x_2\}}) \\
&~~+ \frac{1}{8\cdot7\cdot6}\sum_{\begin{substack}
{\cA(x_1,x_2,x_3)}\end{substack}}
I(W_{\{z_2^{(1)},z_2^{(2)}\}}^{(2)})\\
&~~+ \frac{1}{8\cdot7\cdot6\cdot4}\sum_{\cA(x_1,x_2,x_3,x_4)}
I(W_{\{z_3^{(1)},z_3^{(2)}\}}^{(3)})\\
&\leq \frac{1}{7}\bigg(\sqrt{1-Z_1^2} + 2\sqrt{1-Z_2^2} +
4\sqrt{1-Z_3^2}\bigg)\\
&~~+ \frac{1}{7\cdot6}\bigg(12\sqrt{1-Z_1^2} + 18\sqrt{1-Z_2^2} +
12\sqrt{1-Z_3^2}\bigg)\\
&~~+ \frac{1}{7\cdot6\cdot4}\bigg(96\sqrt{1-Z_1^2} +
48\sqrt{1-Z_2^2} + 24\sqrt{1-Z_3^2}\bigg)\\
&=\sqrt{1-Z_1^2} + \sqrt{1-Z_2^2} + \sqrt{1-Z_3^2}
\end{align*}
This completes the proof of \eqref{eq:upper}.

\def\cprime{$'$}

\end{document}